\documentclass[twoside]{article}

\usepackage[a4paper, hmargin=1in, vmargin={1in, 1.4in}]{geometry}
\usepackage{titling}
\usepackage{amsmath,amssymb,amsthm}
\usepackage{xspace}
\usepackage{multirow}
\usepackage{makecell}
\usepackage{bm}
\usepackage{url}
\usepackage{tikz}
\usepackage[braket]{qcircuit}
\usepackage{calc}

\usepackage[titletoc,title]{appendix}

\allowdisplaybreaks[4]

\newtheorem{mytheorem}{Theorem}
\newtheorem{mylemma}{Lemma}

\theoremstyle{definition}
\newtheorem{mydefinition}[mylemma]{Definition}
\newtheorem{example}[mylemma]{Example}
\newtheorem{construction}{Construction}

\theoremstyle{remark}

\newcommand{\textover}[3][l]{%
 \makebox[\widthof{#3}][#1]{#2}%
 }

\newcommand{\NMQCp}{\ensuremath{\mathrm{NMQC}_\oplus}\xspace}
\newcommand{\ENMQCp}{\ensuremath{\mathsf{ENMQC}_\oplus}\xspace}

\newcommand{\ZZZZ}{\ensuremath{\mathbb{Z}/4\mathbb{Z}}\xspace}
\newcommand{\Z}[1]{\ensuremath{\mathbb{Z}/#1\mathbb{Z}}\xspace}

\newcommand{\ACpz}{\ensuremath{\mathsf{AC_\oplus^0}}\xspace}
\newcommand{\ACCz}{\ensuremath{\mathsf{ACC^0}}\xspace}
\newcommand{\TC}{\ensuremath{\mathsf{TC}}\xspace}
\newcommand{\TCz}{\ensuremath{\mathsf{TC^0}}\xspace}
\newcommand{\NP}{\ensuremath{\mathsf{NP}}\xspace}

\renewcommand{\P}{\ensuremath{\mathsf{P}}\xspace}
\newcommand{\BQP}{\ensuremath{\mathsf{BQP}}\xspace}
\newcommand{\pfs}{\ensuremath{\mathsf{pfs}}\xspace}
\newcommand{\QNCfz}{\ensuremath{\mathsf{QNC_f^0}}\xspace}
\newcommand{\AWPP}{\ensuremath{\mathsf{AWPP}}\xspace}

\title{Periodic Fourier Representation of Boolean Functions}
\author{Ryuhei Mori}
\date{\small \texttt{mori@c.titech.ac.jp}\\ School of Computing, Tokyo Institute of Technology}

\begin{document}
\begin{titlingpage}
\maketitle

\begin{abstract}
In this work, we consider a new type of Fourier-like representation of Boolean function $f\colon\{+1,-1\}^n\to\{+1,-1\}$
\begin{equation*}
f(x) = \cos\left(\pi\sum_{S\subseteq[n]}\phi_S \prod_{i\in S} x_i\right).
\end{equation*}
This representation, which we call the periodic Fourier representation, of Boolean function is closely related to a certain type of multipartite Bell inequalities and non-adaptive measurement-based quantum computation with linear side-processing (\NMQCp).
The minimum number of non-zero coefficients in the above representation, which we call the periodic Fourier sparsity, is equal to the required number of qubits for the exact computation of $f$ by \NMQCp.
Periodic Fourier representations are not unique, and can be directly obtained both from the Fourier representation and the $\mathbb{F}_2$-polynomial representation.
In this work, we first show that Boolean functions related to $\ZZZZ$-polynomial have small periodic Fourier sparsities.
Second, we show that the periodic Fourier sparsity is at least $2^{\deg_{\mathbb{F}_2}(f)}-1$, which means that \NMQCp efficiently computes a Boolean function $f$ if and only if $\mathbb{F}_2$-degree of $f$ is small.
Furthermore, we show that any symmetric Boolean function, e.g., $\mathsf{AND}_n$, $\mathsf{Mod}^3_n$, $\mathsf{Maj}_n$, etc, can be exactly computed by depth-2 \NMQCp using a polynomial number of qubits, that implies exponential gaps between \NMQCp and depth-2 \NMQCp.
\end{abstract}
\end{titlingpage}

\section{Introduction}
\subsection{Periodic Fourier representation}
Fourier analysis of Boolean function is a powerful tool used in theoretical computer science~\cite{odonnell2014analysis}.
A Boolean function $f\colon\{+1,-1\}^n\to\{+1,-1\}$ can be represented by a unique $\mathbb{R}$-multilinear polynomial
\begin{equation}\label{eq:fourier}
f(x) = \sum_{S\subseteq[n]} \widehat{f}(S)\prod_{i\in S} x_i
\end{equation}
using Fourier coefficients $(\widehat{f}(S)\in\mathbb{R})_{S\subseteq[n]}$ where $[n] := \{1,2,\dotsc,n\}$.
Here, the number $\bigl|\{S\in[n]\mid \widehat{f}(S)\ne 0\}\bigr|$ of non-zero Fourier coefficients, called the Fourier sparsity, is one of the important complexity measures of Boolean functions,
which means the number of $\mathbb{F}_2$-linear Boolean functions correlated to $f$.
On the other hand, another natural complexity measure of Boolean function is the linear sketch complexity~\cite{kannan2016linear}, which is the smallest number $k$ such that
there exists a Boolean function $g\colon\{+1,-1\}^k\to\{+1,-1\}$ and $S_1,\dotsc,S_k\subseteq[n]$ satisfying
\begin{equation}\label{eq:lsketch}
f(x) = g\left(\prod_{i\in S_1} x_i, \dotsc, \prod_{i\in S_k}x_i\right).
\end{equation}
In fact, the linear sketch complexity is equal to the Fourier dimension, which is a dimension of a linear space spanned by $\{S\subseteq[n]\mid \widehat{f}(S)\ne 0\}$ where a subset $S$ is regarded as a vector in $\mathbb{F}_2^n$~\cite{montanaro2009ccxor}.
Importantly, the linear sketch complexity has the above operational definition, and also has another operational characterization which is the one-way communication complexity of $f^\oplus(x, y) := f(x\oplus y)$~\cite{montanaro2009ccxor}.

Here, the Fourier representation~\eqref{eq:fourier} can be regarded as a restriction of general linear sketch~\eqref{eq:lsketch} where $g$ must be $\mathbb{R}$-linear.
In this work, we consider a different type of restriction where $g$ must be the cosine function of $\mathbb{R}$-linear function, i.e.,
\begin{equation}\label{eq:tlinear}
f(x) = \cos\left(\pi\sum_{S\subseteq[n]}\phi_S \prod_{i\in S} x_i\right).
\end{equation}
Here, the constant factor $\pi$ is not essential, but introduced for the simplicity of analysis.
We call \eqref{eq:tlinear} the \textit{periodic Fourier representation}.
The number $|\{S\subseteq[n]\mid S\ne \varnothing, \phi_S \ne 0\}|$ of non-zero coefficients except for that corresponding to the empty set is the complexity measure which we will consider in this work,
and call the \textit{periodic Fourier sparsity}.
The periodic Fourier sparsity is operationally characterized as the required number of qubits for computing $f$ exactly by non-adaptive measurement-based quantum computation with linear side-processing (\NMQCp)~\cite{PhysRevA.64.032112}, \cite{hoban2011non}.
This fact is a consequence of Werner and Wolf's theorem~\cite{PhysRevA.64.032112}.
Werner and Wolf showed that for given Boolean function $h\colon \{+1,-1\}^k\to\{+1,-1\}$ and input distribution $\mu$ on $\{+1,-1\}^k$, the largest bias of winning probability of $k$-player XOR game $(h, \mu)$ in quantum theory is equal to
\begin{equation}\label{eq:werner}
\max_{\phi_0,\dotsc,\phi_k}\, \sum_{z_1,\dotsc, z_k} \mu(z_1,\dotsc,z_k) h(z_1,\dotsc,z_k) \cos\left(\pi\left(\phi_0+\sum_{i=1}^k\phi_i z_i \right)\right).
\end{equation}
This largest winning probability is achieved by using shared $k$-qubit GHZ state and local measurements $\cos(\pi(\phi_i z_i+\phi_0/k))X + \sin(\pi(\phi_iz_i+\phi_0/k))Y$ where $X$ and $Y$ are the Pauli matrices~\cite{PhysRevA.64.032112}.
If we assume that $z_1,\dotsc,z_k$ are parities of hidden inputs $x_1,\dotsc,x_n$, then the situation of XOR game is equivalent to that of \NMQCp~\cite{hoban2011non}.
Hence, \NMQCp exactly computes $f$ by using $k$ qubits if and only if the periodic Fourier sparsity of $f$ is at most $k$.

\subsection{Non-adaptive measurement-based quantum computation with linear side-processor}
Measurement-based quantum computation (MBQC) is a model of quantum computation based on qubit-wise measurements of prepared state which is independent of input (actually independent also of problem in the following works).
When measurement outcomes are used for choices of other measurements, we say that the MBQC algorithm is adaptive.
Raussendorf et al.\ showed that MBQC with adaptive measurements and linear side-processing using a cluster state can simulate quantum circuit with small overhead~\cite{PhysRevA.68.022312}.
Hoban et al.\ considered and analyzed non-adaptive MBQC with linear side-processing (\NMQCp)~\cite{hoban2011non}.
Their results will be briefly introduced in the next section.
In the rest of this section, we explain a definition of \NMQCp for computing a Boolean function $f\colon \{0,1\}^n\to\{0,1\}$.
Before an input is given, we prepare a $k$-qubit quantum state $\rho$ for some positive integer $k$ and two binary measurements $A^0_i$ and $A^1_i$ for each qubit indexed by $i\in[k]$.
After an input $\mathsf{x}\in\{0,1\}^n$ is given, the linear side-processor computes $k$ parities $\mathsf{z}_1 := \bigoplus_{i\in S_1} \mathsf{x}_i,\dotsc,\mathsf{z}_k := \bigoplus_{i\in S_k}\mathsf{x}_k$ for fixed subsets $S_1,\dotsc,S_k\subseteq [n]$ which are independent of input $\mathsf{x}$.
Then, $i$-th qubit of $\rho$ is measured by $A^{\mathsf{z}_i}_i$ for each $i\in[k]$ independently.
Finally, the linear side-processor computes a parity of all measurement outcomes, which is the final output of the \NMQCp algorithm and should be equal to $f(\mathsf{x})$.
Hence, \NMQCp algorithm is specified by a positive integer $k$, a prepared $k$-qubit state $\rho$, prepared measurements $A^0_i$ and $A^1_i$ for $i\in[k]$ and subsets $S_1,\dotsc,S_k\subseteq[n]$.

If we consider \NMQCp with the minimum error probability on given input distribution,
Werner and Wolf's theorem implies that we can safely assume that the prepared quantum state $\rho$ is the generalized GHZ state $(\ket{0\dotsm0}+\ket{1\dotsm1})/\sqrt{2}$
and the binary measurements are $\cos(\pi(\phi_i z_i + \phi_0/k))X + \sin(\pi(\phi_i z_i + \phi_0/k)) Y$ for some parameters $\phi_0,\dotsc,\phi_k$ where $z_i = 1-2\mathsf{z}_i$ for $i\in[k]$.
Especially, we can exactly compute a Boolean function $f\colon\{0,1\}^n\to\{0,1\}$ by \NMQCp using $k$ qubits if and only if the periodic Fourier sparsity of $f$ is at most $k$.

\subsection{Background: Foundation of quantum physics by computational complexities}
While quantum physics is described by extremely simple mathematics, quantum physics does not have operationally meaningful axioms (often called ``postulates'' rather than axioms).
Recently, quantum physics is believed to be explained by ``information processing''.
Some postulates based on information processing have been suggested~\cite{PhysRevLett.96.250401}, \cite{pawlowski2009information}, \cite{PhysRevA.84.012311}.
On the other hand, postulates based on ``computational complexity'' have not been investigated sufficiently.
Recently, Barrett et al.\ showed that generalized probabilistic theories, which are ``theories'' including quantum theory and obeying a general framework of theories based on weak assumptions, can solve problems in \AWPP~\cite{barrett2017computational}.
Hence, postulates on computational complexity such as ``Nature does not allow us to solve \NP-hard problem efficiently'' could be a candidate of postulates for quantum physics since some of generalized probabilistic theories violate this postulate unless $\AWPP\subseteq\NP$.
For discussing computations in generalized probabilistic theories, states, measurements and operations have to be defined for multipartite system~\cite{PhysRevA.75.032304}.
On the other hand, in measurement-based computations, we only needs concepts of states and measurements in multipartite system.
Hence, it would be clearer to argue measurement-based computation rather than a standard computation in generalized probabilistic theories since we do not have to define a set of allowed operations in generalized probabilistic theories.

Raussendorf et al.\ showed that adaptive MBQC with linear side-processing using a polynomial-size cluster state can simulate polynomial-size quantum circuit~\cite{PhysRevA.68.022312}.
Anders and Browne observed that adaptive MBQC with linear side-processing using polynomially many tripartite GHZ states can simulate polynomial-size classical circuit~\cite{PhysRevLett.102.050502}.
Raussendorf showed that adaptive measurement-based classical computation can compute only affine Boolean functions~\cite{PhysRevA.88.022322}.
Hoban et al.\ showed that \NMQCp can compute arbitrary Boolean function by using exponentially large generalized GHZ state~\cite{hoban2011non} on the basis of Werner and Wolf's theorem~\cite{PhysRevA.64.032112}.
Furthermore, they showed that the exact computation of $\mathsf{AND}_n$ by \NMQCp requires $2^n-1$ qubits, which means that computational power of efficient \NMQCp is limited.
On the other hand, measurement-based computation with linear side-processing in general no-signaling theory has unlimited computational power since a probability distribution
\begin{equation*}
\Pr(\mathsf{a}_1,\dotsc,\mathsf{a}_n\mid \mathsf{x}_1,\dotsc,\mathsf{x}_n)= \begin{cases}
\frac1{2^{n-1}},&\text{if }  \bigoplus_{i\in[n]} \mathsf{a}_i = f(\mathsf{x})\\
0,&\text{otherwise}
\end{cases}
\end{equation*}
for arbitrary Boolean function $f\colon\{0,1\}^n\to\{0,1\}$ satisfies the no-signaling condition.
Hence, in general no-signaling theory, the required number of ``generalized bits'' in non-adaptive measurement-based computation with linear side-processing for arbitrary Boolean function $f$ is at most $n$, and is equal to the linear sketch complexity of $f$.
Note that here, we do not consider the computational complexity for generating the above state since that is a computation independent of input and can be computed before an input is given.
We are able to argue computational complexity after an input is given.
We may regard this setting as the ``non-uniform setting'', and regard the above prepared state as ``non-classical advice''.
These results are summarized in Table~\ref{tbl:mbc}.
Understanding \NMQCp would be important for characterizing quantum physics since general no-signaling theory with linear side-processing allows us to compute arbitrary Boolean function efficiently.
The periodic Fourier sparsity of Boolean function is equal to the required number of qubits for exact computation by \NMQCp.
Hence, in this work, we mainly investigate efficiencies of exact computations by \NMQCp, which is the part corresponding to ``Restricted (NE)'' in Table~\ref{tbl:mbc}.

\begin{table}[t]
\caption{Classes of Boolean functions computable by measurement-based computation with linear side-processing. A, N, E and P stand for ``adaptive'', ``non-adaptive'', ``exact'' and ``probabilistic (bounded-error)'', respectively.}
\label{tbl:mbc}
\begin{center}
\begin{tabular}{|c|c|c|}
\hline
Theory & Computable & Efficiently computable\\
\hline
Local realistic theory & \multicolumn{2}{c|}{Affine (AP)~\cite{PhysRevA.88.022322}}\\
\hline
Quantum theory & Any (NE)~\cite{hoban2011non} & \makecell{\BQP/\textsf{qpoly} (AP)~\cite{PhysRevA.68.022312}\\ Restricted (NE)~\cite{hoban2011non}}\\
\hline
No-signaling theory & \multicolumn{2}{c|}{Any (NE)}\\
\hline
\end{tabular}
\end{center}
\end{table}

\subsection{Our results}
In this work, we first show some techniques for obtaining periodic Fourier representations~\eqref{eq:tlinear} with small periodic Fourier sparsity on the basis of $\mathbb{R}$-multilinear, $\mathbb{F}_2$-multilinear and $\Z{4}$-multilinear polynomials.
More precisely, Boolean functions with small Fourier sparsity or low $\mathbb{F}_2$-degree have a small periodic Fourier sparsity.
Furthermore, Boolean functions related with \Z{4}-polynomial have a small periodic Fourier sparsity as well.
For instance, the complete quadratic function, $\mathsf{CQ}_n(\mathsf{x}) := \bigoplus_{1\le i < j \le n} \mathsf{x}_i\land \mathsf{x}_j = \lfloor\bigl(\sum_{i\in[n]} \mathsf{x}_i \mod 4\bigr)/2\rfloor$ has the periodic Fourier sparsity $n+1$
on the basis of the $\mathbb{Z}/4\mathbb{Z}$-multilinear polynomial representation
while Fourier representation and $\mathbb{F}_2$-polynomial representation give periodic Fourier sparsities $2^n-1$ and $n(n+1)/2$, respectively.
Currently, we do not know any other method for obtaining a periodic Fourier representation.

Next, we show some lower bounds for the periodic Fourier sparsity of Boolean functions.
Hoban et al.\ showed that the periodic Fourier sparsity of $\mathsf{AND}_n$ is $2^n-1$~\cite{hoban2011non}.
However, $\mathsf{AND}_n$ can be approximated by low ${\mathbb{F}_2}$-degree polynomial including random input variables~\cite{smolensky1987algebraic}, and hence, can be computed efficiently with bounded error by \NMQCp using the above technique.
In this work, we show that the periodic Fourier sparsity of a Boolean function $f$ is at least $2^{\deg_{\mathbb{F}_2}(f)}-1$.
Hence, periodic Fourier sparsities of $\mathsf{Mod}^{3}_n$ and $\mathsf{Maj}_n$ are at least $2^{n-1}-1$.
Since these Boolean functions cannot be approximated by low $\mathbb{F}_2$-degree polynomial~\cite{smolensky1987algebraic},
we can expect that these Boolean functions cannot be computed efficiently by \NMQCp even with bounded error.
Let \ENMQCp be a class of Boolean functions which can be exactly computed by \NMQCp with a polynomial number of qubits, i.e., \ENMQCp is a class of Boolean functions with polynomial periodic Fourier sparsity.
\begin{mytheorem}\label{thm:main}
The periodic Fourier sparsity of Boolean function $f$ is at least $2^{\deg_{\mathbb{F}_2}(f)}-1$.
Hence, $\mathsf{Mod}^3_n$ and $\mathsf{Maj}_n$ are not in \ENMQCp.
\end{mytheorem}

Let \QNCfz be a class of Boolean functions which can be exactly computed by polynomial-size constant-depth quantum circuit with fan-out gates~\cite{green2002counting}, \cite{v001a005}, \cite{Takahashi2016}.
Obviously, $\ENMQCp\subseteq\QNCfz$ as shown in Appendix~\ref{apx:circuit}.
Since $\mathsf{AND}_n, \mathsf{Mod}^3_n, \mathsf{Maj}_n\in\QNCfz$~\cite{Takahashi2016}, we obtain $\ENMQCp\subsetneq\QNCfz$.
It is also easy to see $\ENMQCp\subsetneq\TCz$ where \TCz is a class of Boolean functions which can be computed by polynomial-size constant-depth circuit with $\land,\lnot,\mathsf{Maj}_n$ gates as shown in Appendix~\ref{apx:circuit}.
On the other hand, \QNCfz circuits which exactly compute $\mathsf{AND}_n$, $\mathsf{Mod}^3_n$ and $\mathsf{Maj}_n$ can be directly transformed into ``depth-2'' \NMQCp using a polynomial number of qubits in which outputs of \NMQCp in the first layer are used as inputs of \NMQCp in the second layer
where quantum states used in the first layer and the second layer are not entangled~\cite{v001a005}, \cite{Takahashi2016}.

\begin{mytheorem}\label{thm:nmqc2}
Any symmetric Boolean function, e.g., $\mathsf{AND}_n$, $\mathsf{Mod}^3_n$, $\mathsf{Maj}_n$, etc., can be computed exactly by depth-2 \NMQCp using a polynomial number of qubits.
\end{mytheorem}

Theorem~\ref{thm:nmqc2} shows a significant gap between \NMQCp, which requires exponentially many qubits for $\mathsf{AND}_n$, $\mathsf{Mod}^3_n$ and $\mathsf{Maj}_n$, and depth-2 \NMQCp, which only needs polynomially many qubits for these Boolean functions.
Note that polynomial-depth \NMQCp is not equivalent to general adaptive MBQC since in polynomial-depth \NMQCp, if a measurement outcome of a qubit $q_0$ is used for a measurement choice for a qubit $q_1$, then $q_0$ and $q_1$ must be originally separable.
Theorem~\ref{thm:nmqc2} also implies that constant-depth \NMQCp using a polynomial number of qubits can compute any Boolean functions in \TCz.

Furthermore, since \ACpz, a class of Boolean functions computed by polynomial-size constant-depth circuit using $\land,\oplus,\lnot$ gates, cannot compute the majority function~\cite{smolensky1987algebraic}, \cite{razborov1987lower},
we obtain the following theorem on a weak sampling of \NMQCp.
\begin{mytheorem}\label{thm:acpz}
For any \NMQCp with a polynomial number of qubits which does not necessarily compute some Boolean function exactly, there is generally no \ACpz circuit whose output is in a support of output distribution of the given \NMQCp for any input.
\end{mytheorem}

It has been conjectured that an \ACCz circuit, which is an \ACpz circuit with $\mathsf{Mod}^k_n$ gates for arbitrary fixed integer $k$, cannot compute $\mathsf{Maj}_n$.
If this conjecture is true, Theorem~\ref{thm:acpz} holds also for \ACCz in place of \ACpz.

\subsection{Organization}
Notion and notations used in this paper are introduced in Section~\ref{sec:pre}.
Methods for deriving periodic Fourier representations are shown in Section~\ref{sec:upper}.
In Section~\ref{sec:lower}, we show methods for deriving lower bounds of periodic Fourier sparsity, and show exponential lower bounds for $\mathsf{Mod}^3_n$ and $\mathsf{Maj}_n$.
In Section~\ref{sec:nmqc2}, we show depth-2 \NMQCp algorithms using a polynomial number of qubits computing $\mathsf{AND}_n$, $\mathsf{Mod}^3_n$ and $\mathsf{Maj}_n$.
In Section~\ref{sec:partial}, well-known multipartite Bell inequalities are understood as \NMQCp for partial functions.
Some algebraic techniques useful for multipartite XOR game are shown as well.
In Section~\ref{sec:error}, we generalize the periodic Fourier sparsity for \NMQCp with bounded error, and show a relationship between the number of binary digits of coefficients $(\phi_S)_{S\subseteq[n]}$ and $\mathbb{F}_2$-degree.

\section{Preliminaries}\label{sec:pre}
\subsection{Fourier representation}
Any function $f\colon\{+1,-1\}^n\to\mathbb{R}$ can be uniquely represented by
a $\mathbb{R}$-multilinear polynomial
\begin{equation*}
f(x) = \sum_{S \subseteq[n]} \widehat{f}(S)\prod_{i\in S} x_i
\end{equation*}
where $[n]:=\{1,2,\dotsc,n\}$.
Here, $\widehat{f}(S)$, which is called the Fourier coefficient, satisfies
\begin{equation*}
\widehat{f}(S) = \mathbb{E}\left[f(x)\prod_{i\in S} x_i\right]
:= \frac1{2^n}\sum_{x\in\{+1,-1\}^n} f(x)\prod_{i\in S} x_i.
\end{equation*}
The Fourier sparsity is defined to be the number of non-zero Fourier coefficients.
An $\mathbb{R}$-degree of $f$ is defined by $\mathrm{deg}_{\mathbb{R}}(f) := \max \{|S|\mid S\subseteq[n], \widehat{f}(S)\ne 0\}$.
Let $\mathbb{R}_\ell := \{x \in \mathbb{R}\mid 2^\ell x \in \mathbb{Z}\}$.
Then, for Boolean function $f\colon\{+1,-1\}^n\to\{+1,-1\}$, $\widehat{f}(S) \in \mathbb{R}_{\mathrm{deg}_{\mathbb{R}}(f) -1}$~\cite{odonnell2014analysis}.
The Fourier dimension $\mathrm{dim}(\widehat{f})$ of $f$ is defined by the dimension of linear space on $\mathbb{F}_2$ spanned by $\{1_S \mid S\subseteq[n], \widehat{f}(S) \ne 0\}$
where $1_S\in\mathbb{F}_2^n$ is the vector whose $i$-th element is 1 if and only if $i \in S$.
In this paper, for a binary variable $x\in\{+1,-1\}$, $\mathsf{x}$ denotes the corresponding variable in $\{0,1\}$, i.e., $\mathsf{x} = (1-x)/2$.
Similarly, we sometimes regard a Boolean function as $\{0,1\}^n\to\{0,1\}$ rather than $\{+1,-1\}^n\to\{+1,-1\}$. This is not necessarily explicitly stated.
In this paper, when we consider a Fourier coefficient $\widehat{f}(S)$ of a Boolean function $f$, $f$ is always regarded as $\{+1,-1\}^n\to\{+1,-1\}$.

\subsection{$\mathbb{F}_2$-polynomial representation}
Any Boolean function $f\colon\{0,1\}^n\to\{0,1\}$ can be uniquely represented by
an $\mathbb{F}_2$-multilinear polynomial
\begin{equation*}
f(\mathsf{x}) = \bigoplus_{S\subseteq[n]} c_S \prod_{i\in S} \mathsf{x}_i.
\end{equation*}
Here, a coefficient $c_S$ satisfies
\begin{equation}
c_S = \bigoplus_{\mathsf{x},\, \mathrm{supp}(\mathsf{x})\subseteq S} f(\mathsf{x})
\label{eq:f2c}
\end{equation}
where $\mathrm{supp}(\mathsf{x}):=\{i\in[n]\mid \mathsf{x}_i=1\}$.
An $\mathbb{F}_2$-degree of $f$ is defined by $\deg_{\mathbb{F}_2}(f) := \max\{|S|\mid S\subseteq [n], c_S \ne 0\}$.

\subsection{Periodic Fourier representation}
Any Boolean function $f\colon\{+1,-1\}^n\to\{+1,-1\}$ can be (not uniquely) represented by
\begin{equation*}
f(x) = \cos\left(\pi\sum_{S\subseteq[n]}\phi_S\prod_{i\in S}x_i\right).
\end{equation*}
Here, the number of non-zero coefficients corresponding to non-empty subset $|\{S\subseteq[n]\mid S\ne \varnothing, \phi_S\ne 0\}|$
is called the periodic Fourier sparsity of the representation.
The periodic Fourier sparsity $\pfs(f)$ of $f$ is defined by the minimum of periodic Fourier sparsities of all periodic Fourier representations of $f$.
For a coefficient $\phi_S$, the unique $k$ such that $\phi_S\in\mathbb{R}_k\setminus\mathbb{R}_{k-1}$ is called the number of binary digits of $\phi_S$.
A maximum of the numbers of binary digits of $\phi_S$ for all $S\subseteq[n]$ is called the number of binary digits of a periodic Fourier representation.
Without loss of generality, we can assume that $\phi_S\notin\mathbb{R}_0$ for $S\ne\varnothing$.
Since $\sin(\pi x) = \cos(\pi (x-\frac12))$, the periodic Fourier sparsity and the number of binary digits of non-constant Boolean function are invariant even if a periodic Fourier representation~\eqref{eq:tlinear} uses the sine function in place of the cosine function.
Hence, we will sometimes use~\eqref{eq:tlinear} with the sine function.

\subsection{Specific Boolean functions}
In this section, we will introduce specific Boolean functions which appear in this paper.
The subscript ${}_n$ stands for the number of input variables.
In the following explanations, we assume that Boolean functions are $\{0,1\}^n\to\{0,1\}$.
\begin{itemize}
\item $\mathsf{XOR}_n\colon$ The XOR function.
\item $\mathsf{AND}_n\colon$ The AND function.
\item $\mathsf{OR}_n\colon$ The OR function.
\item $\mathsf{Maj}_n\colon$ The majority function. The number $n$ of input variables is assumed to be odd.
\item $\mathsf{CQ}_n\colon$ The complete quadratic function, i.e., $\mathsf{CQ}_n(\mathsf{x}):=\bigoplus_{1\le i < j \le n} \mathsf{x}_i\land\mathsf{x}_j$.
\item $\mathsf{C}^3_n\colon$ The complete cubic function, i.e., $\mathsf{C}^3_n(\mathsf{x}):=\bigoplus_{1\le i < j < k \le n} \mathsf{x}_i\land\mathsf{x}_j\land\mathsf{x}_k$.
\item $\mathsf{Mod}^k_n\colon$ The $k$-modular counting function, i.e., $\mathsf{Mod}^k_n(\mathsf{x})=1$ if and only if $\sum_{i\in[n]}\mathsf{x}_i$ is divisible by $k$.
\item $\mathsf{Exact}^k_n\colon$ The $k$-exactness function, i.e., $\mathsf{Exact}^k_n(\mathsf{x})=1$ if and only if $\sum_{i\in[n]}\mathsf{x}_i = k$.
\end{itemize}
Furthermore, we define $\mathsf{LSB}^\ell\colon\mathbb{Z}_{\ge 0}\to\{0,1\}$ which is
the $k$-th lowest significant bit (LSB) function, i.e., $\mathsf{LSB}^\ell(m)=1$ if and only if $k$-th LSB of binary representation of $m$ is 1.
Note that $\mathsf{XOR}_n(\mathsf{x}) = \mathsf{LSB}^1\bigl(\sum_{i\in[n]}\mathsf{x_i}\bigr)$ and $\mathsf{CQ}_n(\mathsf{x}) = \mathsf{LSB}^2\bigl(\sum_{i\in[n]}\mathsf{x_i}\bigr)$.

\section{Periodic Fourier sparsity: Upper bounds}\label{sec:upper}
\subsection{$\mathbb{R}$-polynomial and $\mathbb{F}_2$-polynomial}
In this section, we consider how to obtain periodic Fourier representations~\eqref{eq:tlinear} from the Fourier representation and the $\mathbb{F}_2$-polynomial representations.
In contrast to the Fourier representation and the $\mathbb{F}_2$-polynomial representation of Boolean functions, a periodic Fourier representation is not unique.
For instance, $\mathsf{XOR}_2$ and $\mathsf{Maj}_3$ can be represented in the following two ways
\begin{align*}
\mathsf{XOR}_2(x_1,x_2) &= \sin\left(\frac{\pi}2x_1x_2\right)\\
&= \sin\left(\frac{\pi}2\left(-1 + x_1 + x_2\right)\right).\\
\mathsf{Maj}_3(x_1,x_2,x_3) &= \sin\left(\frac{\pi}4\left(x_1+x_2+x_3-x_1x_2x_3\right)\right)\\
&= \sin\left(\frac{\pi}4\left(-1 + 2x_1 + 2x_2 + 2x_3 - x_1x_2 - x_2x_3 - x_3x_1\right)\right).
\end{align*}
While the first representations merely use the Fourier representations,
the second representations use the periodicity of the sine function.
In the following, we show how to generalize these representations to general Boolean functions.
From $f(x) = \sin\left(\frac{\pi}2 f(x)\right) = \sin\left(\frac{\pi}2 \sum_{S\subseteq[n]}\widehat{f}(S) \prod_{i\in S}x_i\right)$, we immediately obtain the following construction.

\begin{construction}\label{const:r}
A Boolean function $f\colon \{+1,-1\}^n \to \{+1,-1\}$ has a periodic Fourier representation
\begin{equation*}
f(x) = \sin\left(\frac{\pi}2 \sum_{S\subseteq[n]}\widehat{f}(S) \prod_{i\in S}x_i\right)
\end{equation*}
with the periodic Fourier sparsity $|\{S\subseteq[n]\mid S\ne\varnothing, \widehat{f}(S)\ne 0\}|$ and the number of binary digits at most $\deg_{\mathbb{R}}(f)$.
\end{construction}

There are Boolean functions whose Fourier sparsity is full but whose $\mathbb{F}_2$-degree is small, e.g., the inner product function, the complete quadratic function of even size, etc~\cite{rothaus1976bent}, \cite{odonnell2014analysis}.
The following construction shows that a Boolean function with low $\mathbb{F}_2$-degree has small periodic Fourier sparsity.

\begin{construction}\label{const:f2}
Assume that Boolean function $f\colon \{0,1\}^n \to \{0,1\}$ has the $\mathbb{F}_2$-polynomial representation
$f(\mathsf{x})= \bigoplus_{S\subseteq[n]} c_S \prod_{i\in S} \mathsf{x}_i$ for $(c_S \in \{0,1\})_{S\subseteq[n]}$.
Then, $f$ has a periodic Fourier representation~\eqref{eq:tlinear}
where
\begin{equation}
\phi_S = (-1)^{|S|}\sum_{T\supseteq S}\frac1{2^{|T|}}c_T
\end{equation}
with the periodic Fourier sparsity $|\{S\subseteq[n]\mid S\ne\varnothing, \exists T\supseteq S, c_T\ne 0\}|\le n^{\mathrm{deg}_{\mathbb{F}_2}(f)}$ and the number of binary digits $\deg_{\mathbb{F}_2}(f)$.
\end{construction}
\begin{proof}
By replacing $\mathsf{x}_i$ with $(1-x_i)/2$, we obtain the real-polynomial representation using modulo 2
\begin{equation*}
f(x) \equiv \sum_{S\subseteq[n]} c_S \prod_{i\in S} \frac{1-x_i}2 \mod 2
\end{equation*}
where $f$ at the left-hand side is a function from $\{+1,-1\}^n$ to $\{0,1\}$.
Hence,
\begin{align*}
f(x) &= \cos\left(\pi\sum_{S\subseteq[n]} c_S \prod_{i\in S} \frac{1-x_i}2\right)\\
&= \cos\left(\pi\sum_{S\subseteq[n]} c_S \frac1{2^{|S|}}\sum_{T\subseteq S}\prod_{i\in T} (-x_i)\right)\\
&= \cos\left(\pi\sum_{T\subseteq[n]} (-1)^{|T|} \left(\sum_{S\supseteq T}  \frac1{2^{|S|}}c_S\right) \prod_{i\in T} x_i\right)
\end{align*}
where $f$ at the left-hand side is a function from $\{+1,-1\}^n$ to $\{+1,-1\}$.
\end{proof}
Here, the number of binary digits obtained by Construction~\ref{const:f2} is optimal.
\begin{mylemma}\label{lem:bdigits}
If $f$ has a periodic Fourier representation~\eqref{eq:tlinear} with $\phi_S\in\mathbb{R}_k$ for all $S\subseteq[n]$, then $\mathrm{deg}_{\mathbb{F}_2}(f)\le k$.
\end{mylemma}
\begin{proof}
For $f\colon\{0,1\}^n\to\{0,1\}$,
there is an integer-valued function $t\colon\{+1,-1\}^n\to\mathbb{Z}$, such that
\begin{equation*}
\sum_{S\subseteq[n]} \phi_S \prod_{i\in S}x_i
= 2 t(x) + f(\mathsf{x}).
\end{equation*}
Hence, from~\eqref{eq:f2c},
\begin{align*}
c_S &\equiv
\sum_{\mathsf{x},\, \mathrm{supp}(\mathsf{x})\subseteq S} \sum_{T\subseteq[n]} \phi_T \prod_{i\in T}(1-2\mathsf{x}_i) \mod 2\\
&\equiv
2^{|S|}\sum_{T\subseteq \bar{S}} \phi_T \mod 2.
\end{align*}
Hence, for $|S|>k$, $c_S=0$.
\end{proof}

From Constructions~\ref{const:r} and \ref{const:f2}, the periodic Fourier representation \eqref{eq:tlinear} has both the features of $\mathbb{R}$-polynomial and $\mathbb{F}_2$-polynomial.
For the complete quadratic function of even size, Constructions~\ref{const:r} and \ref{const:f2} give the periodic Fourier sparsities $2^n-1$ and $n(n+1)/2$, respectively.
Hence, Construction~\ref{const:f2} gives a sparser representation.
Conversely, for some Boolean function, Construction~\ref{const:r} gives a sparser representation.
\begin{example}
For $f(\mathsf{x}) := (\mathsf{x}_1\oplus\dotsb\oplus \mathsf{x}_{n/k})\wedge\dotsm\wedge(\mathsf{x}_{n-n/k+1}\oplus\dotsb\oplus \mathsf{x}_n)$ where $n$ is a multiple of an integer $k$,
Constructions~\ref{const:r} and \ref{const:f2} give the periodic Fourier sparsities $2^k-1$ and $(n/k+1)^k-1$, respectively.
In Section~\ref{sec:lower}, we will show $\pfs(f)\ge 2^{\deg_{\mathbb{F}_2}(f)}-1$, which shows the optimality of Construction~\ref{const:r} in this case.
\end{example}

The following lemma is useful for obtaining a periodic Fourier representation for $f=g\land h$ from periodic Fourier representations for $g$ and $h$.

\begin{mylemma}\label{lem:and}
Let $f_1, f_2,\dotsc, f_k$ be Boolean functions on the common input variables $x_1,\dotsc,x_n$.
Assume $f_j$ has a periodic Fourier representation with a periodic Fourier sparsity $s_j$ and the number of binary digits $\ell_j$.
Then, $\bigwedge_{j=1}^k f_j$ has a periodic Fourier representation with the periodic Fourier sparsity at most $\prod_{j=1}^k(s_j+1)-1$ and the number of binary digits at most $\sum_{j=1}^k \ell_j$.
\end{mylemma}
\begin{proof}
From the assumption, there exists a periodic Fourier representation
\begin{align*}
f_j(x) &= \cos\left(\pi\sum_{S\subseteq[n]} \phi^{(j)}_S \prod_{i\in S} x_i\right)
\end{align*}
using $(\phi_S^{(j)}\in\mathbb{R}_{\ell_j})_{S\subseteq[n]}$ for each $j=1,\dotsc,k$.
Here, $f_j(x)=-1$ if and only if $\sum_{S\subseteq[n]} \phi_S^{(j)}\prod_{i\in S} x_i$ is an odd integer.
Hence,
\begin{align*}
\left(\bigwedge_{j=1}^k f_j\right)(x) &= \cos\left(\pi\prod_{j=1}^k\left(\sum_{S\subseteq[n]} \phi^{(j)}_S \prod_{i\in S} x_i\right)\right).
\end{align*}
\end{proof}
Similarly, for $f=g\oplus h$, we obtain $f(x)=\cos(\pi (\widetilde{g}(x) + \widetilde{h}(x)))$ where $\widetilde{g}$ and $\widetilde{h}$ satisfy $g(x) = \cos(\pi\widetilde{g}(x))$ and $h(x)=\cos(\pi\widetilde{h}(x))$.
This technique would be useful for $f=g\wedge h$ or $f=g\oplus h$ where Construction~\ref{const:r} is suitable for $g$ and Construction~\ref{const:f2} is suitable for $h$.

\subsection{\ZZZZ-polynomial}
There exist Boolean functions whose periodic Fourier sparsity is smaller than those obtained by Constructions~\ref{const:r} and \ref{const:f2}.
\begin{example}\label{exm:cq}
When $n$ is even, the complete quadratic function $\mathsf{CQ}_n$ is a bent function, i.e., $|\widehat{\mathsf{CQ}_n}(S)|=2^{-n/2}$ for all $S\subseteq[n]$~\cite{odonnell2014analysis}.
Hence, Construction~\ref{const:r} gives the periodic Fourier sparsity $2^n-1$.
Construction~\ref{const:f2} gives the periodic Fourier sparsity $n(n+1)/2$.
However, the periodic Fourier sparsity of $\mathsf{CQ}_n$ is smaller.
It is easy to see that $\mathsf{CQ}_n(\mathsf{x})$ is the second LSB in the number of 1s in $\mathsf{x}$ as shown in Appendix~\ref{apx:lsb}.
Hence, we obtain the following periodic Fourier representation
\begin{equation*}
\mathsf{CQ}_n(x) = \cos\left(\frac{\pi}2\left(\sum_{i=1}^n \frac{1-x_i}2 - \frac{1-\prod_{i=1}^n x_i}2\right)\right)
\end{equation*}
with the periodic Fourier sparsity $n+1$ and the number of binary digits $2$.
This explains the result in~\cite{hoban2011non} using the periodic Fourier representation.
\end{example}

We can generalize Example~\ref{exm:cq} as follows.
\begin{construction}\label{const:z4f}
Assume that a Boolean function $f\colon \{0,1\}^n\to\{0,1\}$ has a representation
\begin{equation*}
f(\mathsf{x}) = \mathsf{LSB}^2\left(\sum_{S\subseteq[n], |S|\le k} c_S \prod_{i\in S} \mathsf{x}_i\right)
\end{equation*}
using integers $(c_S\in\{0,1,2,3\})_S$.
Let
\begin{equation*}
g(\mathsf{x}) := \sum_{S\subseteq[n], |S|\le k} c_S \prod_{i\in S} \mathsf{x}_i \mod 2.
\end{equation*}
Then, $f$ has a periodic Fourier representation
\begin{equation*}
f(x) = \cos\left(\frac{\pi}2\left(\sum_{S\subseteq[n],|S|\le k} c_S \prod_{i\in S}\frac{1-x_i}2 - \frac{1-\sum_{S\subseteq[n]}\widehat{g}(S)\prod_{i\in S} x_i}2\right)\right)
\end{equation*}
with the periodic Fourier sparsity at most $|\{S\subseteq[n]\mid S\ne \varnothing,\, \exists T\supseteq S, c_T \ne 0 \}|+|\{S\subseteq[n]\mid S\ne\varnothing, \widehat{g}(S)\ne 0\}|$
and the number of binary digits at most $\max\{k, \mathrm{deg}_{\mathbb{R}}(g)\}+1$.
\end{construction}

\begin{example}\label{exm:c3}
The complete cubic function $\mathsf{C}^3_n(\mathsf{x})$ is 1 if $\sum_{i=1}^n \mathsf{x}_i \equiv 3 \mod 4$, and 0 otherwise.
From $\mathsf{C}^3_n = \mathsf{CQ}_n\land \mathsf{XOR}_n$, Example~\ref{exm:cq} and the construction in Lemma~\ref{lem:and},
we obtain a periodic Fourier representation
\begin{align*}
\mathsf{C}^3_n(x) &= \cos\left(\frac{\pi}2\left(\sum_{i=1}^n\frac{1-x_i}2 - \frac{1-\prod_{i=1}^nx_i}2\right)\frac{1-\prod_{i=1}^nx_i}2\right)\\
&= \cos\left(\frac{\pi}8\left(n - 2 - \sum_{i=1}^nx_i \right)\left(1-\prod_{i=1}^nx_i\right)\right)
\end{align*}
with a periodic Fourier sparsity $2n+1$ (or $2n$ when $n\equiv 2 \mod 8$) and the number of binary digits 3 for $n\ge 3$.
\end{example}

Currently, we do not know any Boolean function whose periodic Fourier sparsity is not given by Constructions~\ref{const:r}, \ref{const:f2}, and \ref{const:z4f} or their combination by Lemma~\ref{lem:and}.
Note that Construction~\ref{const:z4f} cannot be generalized to the third LSB since the Fourier representation for the second LSB has high Fourier sparsity and cannot be used for the cancellation
(we cannot use periodic Fourier representations of $\mathsf{LSB}^2_n$ using the periodicity of the cosine function for the cancellation).

\section{Periodic Fourier sparsity: Lower bounds}\label{sec:lower}

In this section, we show lower bounds of the periodic Fourier sparsity of given Boolean function.
Gopalan et al.\ showed a relationship between the Fourier sparsity and the number of binary digits of the Fourier coefficients~\cite{gopalan2011testing}.
This result can be straightforwardly generalized to the periodic Fourier representation.
\begin{mylemma}[\cite{gopalan2011testing}]\label{lem:gran}
For any Boolean function $f\colon \{+1,-1\}^n\to\{+1,-1\}$ with a periodic Fourier representation with the periodic Fourier sparsity $s$, all coefficients $\phi_S$ in the representation are in $\mathbb{R}_{\lfloor\log(s+1)\rfloor}$.
\end{mylemma}
\begin{proof}
We will prove this lemma by induction on $n$.
For $n=0$, the lemma obviously holds.
For $n\ge 1$, we consider two cases.
First, we assume $s=2^n-1$.
There exists some integer-valued function $t\colon\{+1,-1\}^n\to\mathbb{Z}$ such that
\begin{equation*}
\sum_{S\subseteq[n]} \phi_S \prod_{i\in S}x_i
= 2 t(x) + \frac{1-f(x)}2.
\end{equation*}
For $S\subseteq[n]$,
\begin{equation*}
\phi_S
= 2 \widehat{t}(S) - \frac{\widehat{f}(S)}2 + \frac12\delta_S
\end{equation*}
where $\delta_S$ is 1 if $S$ is the empty set, and is 0 otherwise.
Here, $\widehat{t}(S)\in\mathbb{R}_n$ since $t(x)$ is an integer-valued function.
Since $\widehat{f}(S)\in\mathbb{R}_{n-1}$, we obtain $\phi_S\in\mathbb{R}_n$ for $n\ge 1$.

Second, we assume $s<2^n-1$. In this case, there exists a non-empty $S^*\subseteq[n]$ such that $\phi_{S^*}=0$.
We will show $\phi_U\in\mathbb{R}_{\lfloor\log(s+1)\rfloor}$ for arbitrary $U\subseteq[n]$ not equal to $S^*$.
Let $i^*\in[n]$ be an index included only by one of $S^*$ and $U$.
Let $V:=(S^*\oplus U)\setminus\{i^*\}$ where $\oplus$ stands for the symmetric difference of two sets. Let $h\colon\{+1,-1\}^{n-1}\to\{+1,-1\}$ be
\begin{equation*}
h(x_1,\dotsc,x_{i^*-1},x_{i^*+1},\dotsc,x_n) := f\left(x_1,\dotsc,x_{i^*-1}, \prod_{i\in V} x_i, x_{i^*+1}, \dotsc, x_n\right).
\end{equation*}
We can straightforwardly obtain a periodic Fourier representation of $h$ from that of $f$ by replacing $x_{i^*}$ with $\prod_{i\in V} x_i$.
In this transformation, two terms for $S\subseteq[n]\setminus\{i^*\}$ and $(S\oplus V) \cup \{i^*\}$ are merged into a single term.
In the above transform, $\phi_{S^*}=0$ and $\phi_U$ are merged, which means that $\phi_U$ is still one of the coefficient in the periodic Fourier representation of $h$.
The periodic sparsity of the representation of $h$ is obviously at most $s$.
Hence, from the induction hypothesis, we obtain $\phi_U \in\mathbb{R}_{\lfloor\log(s+1)\rfloor}$.
\end{proof}
From Lemmas~\ref{lem:bdigits} and \ref{lem:gran}, we obtain Theorem~\ref{thm:main}, i.e., $\pfs(f) \ge 2^{\deg_{\mathbb{F}_2}(f)}-1$.
For $\mathsf{Mod}^3_n$, this lower bound matches the upper bound obtained by Construction~\ref{const:r}.
\begin{mylemma}
For $n$ divisible by 3, $\pfs(\mathsf{Mod}^3_n) = 2^{n-1}-1$.
For $n$ not divisible by 3, $\pfs(\mathsf{Mod}^3_n) = 2^{n}-1$.
\end{mylemma}
\begin{proof}
From~\eqref{eq:f2c}, it is easy to see that $\deg_{\mathbb{F}_2}(\mathsf{Mod}^3_n)$ is equal to $n-1$ if $n$ is divisible by 3, and is equal to $n$ if $n$ is not divisible by 3.
Hence, from Theorem~\ref{thm:main}, we obtain the lower bounds.
From Construction~\ref{const:r} and the Fourier representation shown in Appendix~\ref{apx:mod3}, we obtain the upper bounds.
\end{proof}

Finally, we show the following lower bound, which is at most $n+1$ but useful for showing the optimality of Example~\ref{exm:cq}. 
\begin{mylemma}
For a Boolean function $f\colon \{+1,-1\}^n\to\{+1,-1\}$ with $\deg_{\mathbb{F}_2}(f) \ge 2$, $\pfs(f)\ge \mathrm{dim}(\widehat{f})+1$.
\end{mylemma}
\begin{proof}
$\pfs(f)\ge \mathrm{dim}(\widehat{f})$ since the Fourier dimension is the linear sketch complexity~\cite{montanaro2009ccxor}.
We assume $\pfs(f)= \mathrm{dim}(\widehat{f})$. Then, all monomials in a periodic Fourier representation~\eqref{eq:tlinear} are linearly independent.
Hence, we can control each term independently, that implies all non-zero coefficients $\phi_S$ must be in $\mathbb{R}_1$.
In that case, $f$ must be an affine function, i.e., $\deg_{\mathbb{F}_2}(f)\le 1$.
\end{proof}

\section{Depth-2 \NMQCp algorithms for $\mathsf{AND}_n$, $\mathsf{Mod}^3_n$ and $\mathsf{Maj}_n$}\label{sec:nmqc2}
In Section~\ref{sec:lower}, we showed that $\mathsf{AND}_n$, $\mathsf{Mod}^3_n$ and $\mathsf{Maj}_n$ cannot be exactly computed by \NMQCp using a polynomial number of qubits.
Interestingly, \QNCfz circuits which exactly compute the above Boolean functions can be directly transformed into ``depth-2'' \NMQCp~\cite{v001a005}, \cite{Takahashi2016}.
\begin{mydefinition}[Depth-$d$ \NMQCp]
Depth-$d$ \NMQCp consists of $d$ layers of \NMQCp.
Qubits used in the same layers could be entangled. However, qubits used in different layers have to be separable.
At the first layer, qubits are locally measured according to $\mathbb{F}_2$-linear functions of input $x$.
At $i$-th layer, qubits are locally measured according to $\mathbb{F}_2$-linear functions of input $x$ and outcomes $a_1,\dotsc,a_{i-1}$ of the previous layers for $i\in[d]$ where $a_i$ denotes the outcomes of the local measurements at $i$-th layer.
An output of depth-$d$ \NMQCp is an $\mathbb{F}_2$-linear function of all outcomes $a_1,\dotsc,a_d$ of the local measurements.
\end{mydefinition}
\begin{proof}[Proof of Theorem~\ref{thm:nmqc2}]
We first sketch the depth-2 \NMQCp algorithm for $\mathsf{OR}_n$ function~\cite{v001a005}, \cite{Takahashi2016}.
For each $k\in \{0,1,\dotsc, \lfloor \log n \rfloor\}$, \NMQCp can compute $Z_k\in\{+1,-1\}$ with expectation
\begin{equation}\label{eq:hoyer}
\mathbb{E}[Z_k] = \cos\left(\frac{\pi}{2^k}\sum_{i\in[n]} \frac{1-x_i}2\right)
\end{equation}
by using $n$ qubits.
If $\mathsf{x}$ is all-zero, $Z_k=+1$ with probability 1 for all $k$.
If $\mathsf{x}$ is not all-zero, $\sum_{i\in[n]}\mathsf{x}_i=2^{k^*} h$ for some positive integer $k^*$ and a positive odd integer $h$.
Hence, $Z_{k^*}=-1$ with probability 1.
The above \NMQCp algorithm reduces $\mathsf{OR}_n$ function to $\mathsf{OR}_{\lfloor \log n\rfloor +1}$ function.
Then, we can apply the \NMQCp algorithm using exponentially many qubits to $\mathsf{OR}_{\lfloor \log n\rfloor +1}$.
The total number of qubits used in the depth-2 \NMQCp algorithm is $(\lfloor \log n\rfloor + 1) n + 2^{\lfloor \log n \rfloor + 1} - 1$.
By introducing an appropriate constant term in~\eqref{eq:hoyer}, we obtain a depth-2 \NMQCp algorithm for $\mathsf{Exact}^k_n$ using the same number of qubits.
By taking a parity of $\mathsf{Exact}^k_n$ for appropriate $k$s, we obtain depth-2 \NMQCp algorithms for arbitrary symmetric Boolean function including $\mathsf{Mod}^3_n$ and $\mathsf{Maj}_n$.
\end{proof}
Hence, there is an exponential gap between \NMQCp and depth-2 \NMQCp.
From the above depth-2 \NMQCp algorithm for $\mathsf{Maj}_n$, we obtain Theorem~\ref{thm:acpz} since \ACpz cannot compute the majority function~\cite{smolensky1987algebraic}, \cite{razborov1987lower}.
Conversely, \NMQCp with a quasi-polynomially many qubits can simulate \ACpz circuit with bounded error by Construction~\ref{const:f2} since \ACpz circuit can be approximated by an $\mathbb{F}_2$-polynomial including random input variables of $\mathbb{F}_2$-degree
$\mathsf{poly}(\log \frac{n}{\epsilon})$ with error probability at most $\epsilon$~\cite{smolensky1987algebraic}.
Theorem~\ref{thm:nmqc2} implies that constant-depth \NMQCp using a polynomial number of qubits can compute Boolean functions in \TCz.
However, it is an open question whether constant-depth \NMQCp using a polynomial number of qubits can compute Boolean functions not in \TCz.
If qubits used in different layers are allowed to be entangled, constant-depth measurement-based quantum computation has the same computational power as \QNCfz~\cite{browne2010computational}.

\section{Partial function and multipartite Bell inequalities}\label{sec:partial}
\begin{figure}[t]
\centering
\hfill
\begin{tikzpicture}[scale=2,label distance=0]
\draw[->] (-1.2,0) -- (1.2,0);
\draw[->] (0,-1.2) -- (0,1.2);
\draw (0, 0) circle[radius=1];
\foreach \x in {0,...,9}
  \draw (0, 0) -- (360*\x/10:1) node [label={[label distance=2mm]360*\x/10:\texttt{\x}}] {};
\foreach \x in {0,...,9}
  \filldraw (360*\x/10:1) circle[radius=0.03];
\end{tikzpicture}
\hfill
\begin{tikzpicture}[scale=2,label distance=0]
\draw[->] (-1.2,0) -- (1.2,0);
\draw[->] (0,-1.2) -- (0,1.2);
\draw (0, 0) circle[radius=1];
\draw (0, 0) -- (45:1) node [label=45:\texttt{01}] {};
\draw (0, 0) -- (135:1) node [label=135:\texttt{10}] {};
\draw (0, 0) -- (-135:1) node [label=-135:\texttt{11}] {};
\draw (0, 0) -- (-45:1) node [label=-45:\texttt{00}] {};
\filldraw (45:1) circle[radius=0.03];
\filldraw (135:1) circle[radius=0.03];
\filldraw (-135:1) circle[radius=0.03];
\filldraw (-45:1) circle[radius=0.03];
\end{tikzpicture}
\hfill {}

\caption{
Left: Generalized the GHZ--Mermin paradox for $P^5_n$.
Right: Maximum violation of Svetlichny's inequality by Belinski\u{\i} and Klyshko.
}
\label{fig:ghz}
\end{figure}
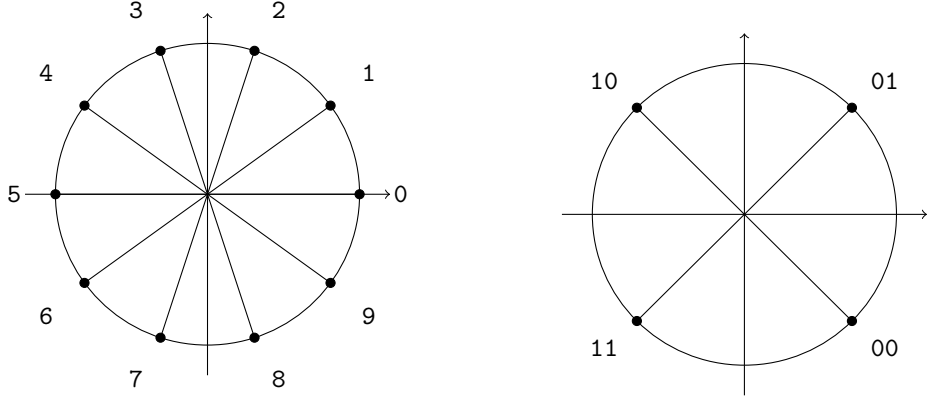

In this section, we briefly introduce some partial functions which can be exactly computed by \NMQCp with $n$ qubits.
The following examples have been known in the context of multipartite Bell inequalities. However, to the knowledge of the author, intuitive graphical interpretations
by angles on a unit circle using Werner and Wolf's characterization~\eqref{eq:werner} have not been shown.
A partial Boolean function $P^k_n\colon\{+1,-1\}^n\to\{+1,-1\}$ is defined by
\begin{equation*}
P^k_n(x) :=\begin{cases}
+1,& \text{if } \sum_{i\in[n]}\mathsf{x}_i \equiv 0 \mod 2k\\
-1,& \text{if } \sum_{i\in[n]}\mathsf{x}_i \equiv k \mod 2k.
\end{cases}
\end{equation*}
This partial Boolean function can be represented by
\begin{equation*}
P^k_n(x) = \cos\left(\frac{\pi}{k}\sum_{i\in[n]} \frac{1-x_i}2\right).
\end{equation*}
This is a simple generalization of the GHZ--Mermin paradox~\cite{Brassard2005}.
Note that this idea was used in~\eqref{eq:hoyer} in which there is $k$ on which the above promise is satisfied.
This idea is useful for total functions with bounded error and for multipartite Bell inequalities as well.
The complete quadratic function $\mathsf{CQ}_n(x)$, which is the second LSB of $\sum_{i\in[n]}\mathsf{x}_i$,
satisfies
\begin{equation*}
\frac1{\sqrt{2}}\mathsf{CQ}_n(x) = \cos\left(\frac{\pi}2\left(-\frac12+\sum_{i\in[n]} \frac{1-x_i}2\right)\right).
\end{equation*}
This argument quite simply explains Belinski\u{\i} and Klyshko's maximum quantum violation of Svetlichny's inequality, which is a generalization of maximum quantum violation of CHSH inequality~\cite{PhysRevA.77.032120}.
Fig.~\ref{fig:ghz} shows graphical interpretations of the above quantum algorithms.

The complete quadratic function $\mathsf{CQ}_n$ is useful for computing general XOR functions distributively.
Multipartite XOR game for XOR function is an important problem in the context of foundations of quantum physics~\cite{PhysRevLett.96.250401}, \cite{PhysRevA.77.032120}, \cite{PhysRevA.94.052130}.
An $n$-partite distributive AND function can be computed by
\begin{align*}
\mathsf{AND}_2\left(\bigoplus_{i=1}^n \mathsf{x}_1^i, \bigoplus_{i=1}^n \mathsf{x}_2^i\right) &= \mathsf{CQ}_{2n}(\mathsf{x}^1_1,\dotsc,\mathsf{x}^n_2) \oplus \mathsf{CQ}_n(\mathsf{x}^1_1,\dotsc,\mathsf{x}^n_1)\\
&\quad \oplus \mathsf{CQ}_n(\mathsf{x}^1_2,\dotsc,\mathsf{x}^n_2)
\end{align*}
where $\mathsf{x}^i_1$ and $\mathsf{x}^i_2$ are inputs for $i$-th player for $i\in[n]$.
Hence, any $\mathbb{F}_2$-quadratic Boolean function can be distributively computed by using $\mathsf{CQ}_n$, e.g.,
the majority function on 3 bits $\mathsf{Maj}_3(\mathsf{x},\mathsf{y},\mathsf{z})=\mathsf{CQ}_3(\mathsf{x},\mathsf{y},\mathsf{z})=\mathsf{x}\land \mathsf{y}\oplus \mathsf{y}\land \mathsf{z}\oplus \mathsf{z}\land \mathsf{x}$
can be distributively computed by
\begin{align*}
\mathsf{Maj}_3\left(\bigoplus_{i=1}^n \mathsf{x}^i, \bigoplus_{i=1}^n \mathsf{y}^i, \bigoplus_{i=1}^n \mathsf{z}^i\right)
&=
\mathsf{CQ}_{2n}(\mathsf{x}, \mathsf{y}) \oplus \mathsf{CQ}_n(\mathsf{x}) \oplus \mathsf{CQ}_n(\mathsf{y})\\
&\quad \oplus\mathsf{CQ}_{2n}(\mathsf{y}, \mathsf{z}) \oplus \mathsf{CQ}_n(\mathsf{y}) \oplus \mathsf{CQ}_n(\mathsf{z})\\
&\quad \oplus\mathsf{CQ}_{2n}(\mathsf{z}, \mathsf{x}) \oplus \mathsf{CQ}_n(\mathsf{z}) \oplus \mathsf{CQ}_n(\mathsf{x})\\
&=\mathsf{CQ}_{2n}(\mathsf{x}, \mathsf{y}) \oplus \mathsf{CQ}_{2n}(\mathsf{y}, \mathsf{z}) \oplus \mathsf{CQ}_{2n}(\mathsf{z}, \mathsf{x}).
\end{align*}
Hence, $\mathsf{Maj}_3$ can be distributively computed with bias $2^{-3/2}$ in quantum theory.
This argument explains techniques in~\cite{PhysRevA.77.032120} in the algebraic way.
As another example, an $n$-partite distributive complete quadratic function can be computed by
\begin{align*}
\mathsf{CQ}_{m}\left(\bigoplus_{i=1}^n \mathsf{x}^i_1, \bigoplus_{i=1}^n \mathsf{x}^i_2, \dotsb, \bigoplus_{i=1}^n \mathsf{x}^i_{m}\right)
&=
\mathsf{CQ}_{nm}(\mathsf{x}) \oplus \bigoplus_{j=1}^m\mathsf{CQ}_n(\mathsf{x}_j)
\end{align*}
with bias $2^{-(m+1)/2}$ in quantum theory.

\section{Discussions on \NMQCp with bounded error}\label{sec:error}
Obviously, bounded-error computational complexities are much more important than zero-error computational complexities in general.
We can consider two models of bounded-error \NMQCp.
In the first setting, we consider the best \NMQCp algorithm for given input distribution.
In this case, from Werner and Wolf's theorem~\eqref{eq:werner}, it is sufficient to consider Werner and Wolf's \NMQCp algorithms, which use shared generalized GHZ state and particular type of local measurements.
The required number of qubits is represented by the following approximate periodic Fourier sparsity.

\begin{mydefinition}
For a Boolean function $f\colon\{+1,-1\}^n\to\{+1,-1\}$, an input distribution $\mu$ on $\{+1,-1\}^n$ and an error probability $\epsilon\in[0,1/2)$,
an approximate periodic Fourier sparsity $\widetilde{\pfs}_{\mu,\epsilon}(f)$ is defined by
\begin{equation*}
\widetilde{\pfs}_{\mu,\epsilon}(f) :=
\min\left\{\pfs(g) \mid
g\colon\{+1,-1\}^n\to[-1,+1],\,
\mathbb{E}_{x\sim\mu}[|f(x)-g(x)|] \le 2\epsilon\right\}.
\end{equation*}
\end{mydefinition}

Let $F\colon\{+1,-1\}^n\to\{+1,-1\}$ be a probabilistic Boolean function such that $F(x)=+1$ with probability $(1+g(x))/2$.
Since $|f(x) - g(x)| = 1 - f(x)g(x)$, $\Pr_{F, x\sim \mu}(F(x) \ne f(x)) = \mathbb{E}_{F,x\sim\mu}[(1-F(x)f(x))/2] = \mathbb{E}_{x\sim\mu}[(1-g(x)f(x))/2] = \mathbb{E}_{x\sim\mu}[|f(x)-g(x)|/2]\le\epsilon$.
Hence, $\widetilde{\pfs}_{\mu,\epsilon}(f)$ is equal to the required number of qubits for computing $f$ with error probability at most $\epsilon$ on input distribution $\mu$.
From the minimax theorem, $\max_\mu \widetilde{\pfs}_{\mu,\epsilon}(f)$ is equal to the number of required qubits for probabilistic \NMQCp algorithm computing $f$ with error probability at most $\epsilon$ for any input.
Here, a probabilistic \NMQCp algorithm has to randomly choose a set of parities used in the computation.
Hence, this computational model does not necessarily have a deterministic linear side-processor.
In the second setting, we consider (not probabilistic) general \NMQCp algorithms with error probability at most $\epsilon$ for any input.
For a fixed set $\mathcal{T}\subseteq 2^{[n]}\setminus\{\varnothing\}$ of parities used in \NMQCp, we can apply the minimax theorem.
Hence, there is an \NMQCp algorithm using parities in $\mathcal{T}$ which computes $f$ with error probability at most $\epsilon$ for any input
if and only if there exist random variables $(\Phi_S\in\mathbb{R})_{S\in\mathcal{T}\cup\{\varnothing\}}$ such that
\begin{equation}\label{eq:apfs}
\left|f(x) - \mathbb{E}_{\Phi}\left[\cos\left(\Phi_\varnothing+\sum_{S\in\mathcal{T}}\Phi_S \prod_{i\in S}x_i\right)\right]\right| \le 2\epsilon
\end{equation}
for any $x\in\{+1,-1\}^n$.
Note that a corresponding \NMQCp algorithm uses generalized GHZ state and local measurements $\mathbb{E}_\Phi[\cos(\pi(\Phi_S\prod_{i\in S} x_i + \Phi_\varnothing / |\mathcal{T}|))X + \sin(\pi(\Phi_S\prod_{i\in S} x_i + \Phi_\varnothing / |\mathcal{T}|))Y]$
for $S\in\mathcal{T}$.
The required number of qubits is represented by the following approximate periodic Fourier sparsity.

\begin{mydefinition}
For a Boolean function $f\colon\{+1,-1\}^n\to\{+1,-1\}$ and an error probability $\epsilon\in[0,1/2)$,
an approximate periodic Fourier sparsity $\widetilde{\pfs}_{\epsilon}(f)$ is defined by minimum $|\mathcal{T}|$
among all $\mathcal{T}\subseteq 2^{[n]}\setminus\{\varnothing\}$ satisfying~\eqref{eq:apfs}.
\end{mydefinition}

From the above argument, $\max_\mu \widetilde{\pfs}_{\mu,\epsilon}(f) \le \widetilde{\pfs}_\epsilon(f)$.
Note that in the context of query complexity, approximate $\mathbb{R}$-degree with respect to the infinity norm plays a similar role of approximate Fourier sparsities although it only gives a lower bound of quantum query complexity~\cite{beals2001quantum}.
Deriving lower bounds of these approximate periodic Fourier sparsities is an open problem.
A connection between the number of binary digits of approximate periodic Fourier representation and $\mathbb{F}_2$-degree of $f$ can be obtained similarly to Lemma~\ref{lem:bdigits}.
\begin{mytheorem}\label{thm:lower}
Assume that a Boolean function $f\colon\{+1,-1\}^n\to\{+1,-1\}$ satisfying \eqref{eq:apfs} for $\mathcal{T}=2^{[n]}\setminus\{\varnothing\}$ using random variables $(\Phi_S\in\mathbb{R}_\ell)_{S\subseteq[n]}$.
Then, there is a probabilistic polynomial $p$ of $\mathbb{F}_2$-degree at most $2^{\ell}-1$ satisfying $\Pr_p(p(x) \ne f(x))\le \epsilon$ for any $x\in\{+1,-1\}^n$.
\end{mytheorem}
\begin{proof}
We will construct a probabilistic polynomial of $\mathbb{F}_2$-degree at most $2^\ell-1$ which is equal to 0 with probability exactly equal to $(1+\cos(\pi\sum_{S\subseteq[n]}\phi_S\prod_{i\in S}x_i)/2$ for each realization $(\phi_S)_{S\subseteq[n]}$ of $(\Phi_S)_{S\subseteq[n]}$.
For each $S\subseteq[n]$, we define $\bar{\phi}_S := 2^\ell \phi_S$, which is guaranteed to be an integer from the assumption.
$\cos\bigl(\pi\sum_{S\subseteq[n]} \phi_S \prod_{i\in S} x_i\bigr)$ is equal to
$\cos\bigl(\pi\bigl(-\sum_{S\subseteq[n]} \phi_S \prod_{i\in S} x_i + 2k\bigr)\bigr)$ for any integer $k$.
Let $k = \lceil \frac12\sum_{S\subseteq[n]} \phi_S\rceil$ so that $2k - \sum_{S\subseteq[n]} \phi_S\ge 0$.
Then, $\cos\bigl(\pi\bigl(\sum_{S\subseteq[n]} \phi_S \prod_{i\in S} x_i\bigr)\bigr)$ is determined by
$y_i(x) := \mathsf{LSB}^{i+1}\left(-\sum_{{S\subseteq[n]}} \bar{\phi}_S \prod_{i\in S} x_i + k 2^{\ell+1}\right)$ for $i=0,1,\dotsc,\ell$.
Obviously, $y_0$ is a constant function.
For $i \ge 1$,
\begin{align*}
y_i(x)&=\mathsf{LSB}^{i+1}\left(-\sum_{S\subseteq[n]} \bar{\phi}_S \prod_{i\in S} x_i + k2^{\ell+1}\right)
=
\mathsf{LSB}^{i+1}\left(\sum_{S\subseteq[n]} \bar{\phi}_S \left(2\bigoplus_{i\in S} \mathsf{x}_i - 1\right) + k2^{\ell+1}\right)\\
&=\mathsf{LSB}^{i+1}\left(2\sum_{S\subseteq[n]} \bar{\phi}_S \bigoplus_{i\in S} \mathsf{x}_i + \left(k2^{\ell+1}- \sum_{S\subseteq[n]} \bar{\phi}_S\right)\right)\\
&=\mathsf{LSB}^{i}\left(\sum_{S\subseteq[n]} \bar{\phi}_S \bigoplus_{i\in S} \mathsf{x}_i + \left\lfloor\frac{k2^{\ell+1}- \sum_{S\subseteq[n]} \bar{\phi}_S}2\right\rfloor\right).
\end{align*}
We can obtain an explicit $\mathbb{F}_2$-polynomial representation of $y_i(x)$ by using the following lemma.

\begin{mylemma}\label{lem:lsb}
For $\ell\ge 1$,
\begin{equation*}
\mathsf{LSB}^{\ell}\left(\sum_{i\in[n]} \mathsf{x}_i\right)
= 
\bigoplus_{1\le i_1<i_2<\dotsb<i_{2^{\ell-1}}\le n} \bigwedge_{j=1}^{2^{\ell-1}} \mathsf{x}_{i_j}.
\end{equation*}
\end{mylemma}
The proof is shown in Appendix~\ref{apx:lsb}.
From Lemma~\ref{lem:lsb},
$y_i(x)$ has $\mathbb{F}_2$-degree at most $2^{i-1}$ for $i\ge 1$.
Then, $\cos(\sum_{S\subseteq[n]} \phi_S \prod_{i\in S} x_i) = \cos(\sum_{i=0}^\ell 2^{-i}y_{\ell-i}(x))$.
Let $Z$ be a random variable uniformly distributed in $[0,1]$ which is a randomness used in the probabilistic polynomial.
Let
\begin{equation*}
S :=
\begin{cases}
0,& \text{if } \frac{1+\cos\left(\pi\sum_{i=0}^\ell 2^{-i} y_{\ell-i}(x)\right)}2 > Z\\
1,& \text{otherwise.}
\end{cases}
\end{equation*}
Then, $S$ is equal to 0 with probability
\begin{equation*}
\frac{1+\cos\left(\pi\sum_{i=0}^\ell 2^{-i} y_{\ell-i}\right)}2
=
\frac{1+\cos\left(\pi\sum_{S\subseteq[n]}\phi_S \prod_{i\in S} x_i\right)}2
\end{equation*}
Since $y_i(x)$ have $\mathbb{F}_2$-degree at most $2^{i-1}$ for $i\in[\ell]$,
a probabilistic polynomial computing $S$ has $\mathbb{F}_2$-degree at most $2^{\ell-1} + 2^{\ell-2}+\dotsb +2^0 = 2^{\ell}-1$.
By using the randomness of $(\Phi_S)_{S\subseteq[n]}$,
we obtain a probabilistic polynomial of $\mathbb{F}_2$-degree at most $2^\ell - 1$ with error probability at most $\epsilon$.
\end{proof}
Hence, a Boolean function which can be computed with bounded error by \NMQCp with constant number of binary digits can be approximated by probabilistic $\mathbb{F}_2$-polynomial of constant degree.

\section*{Acknowledgment}
This work was supported by JST PRESTO Grant Number JPMJPR1867 and JSPS KAKENHI Grant Number JP17K17711.

\section*{References}
\renewenvironment{thebibliography}[1]
        {\frenchspacing
         \small\rm\baselineskip=11pt
         \begin{list}{\arabic{enumi}.}
        {\usecounter{enumi}\setlength{\parsep}{0pt}     
         \setlength{\leftmargin}{17pt}  
                \setlength{\rightmargin}{0pt}
         \setlength{\itemsep}{0pt} \settowidth
          {\labelwidth}{#1.}\sloppy}}{\end{list}}

\bibliographystyle{plain}

\begin{appendices}
\section{Non-adaptive measurement-based quantum computation and circuit model}\label{apx:circuit}
Here, we give a graphical proof that~\eqref{eq:werner} can be achieved by generalized GHZ state and local measurements.
Let $H$ be the $2\times 2$ Hadamard transform and $R_z(\theta) := \ket{0}\bra{0} + \mathrm{e}^{i\theta}\ket{1}\bra{1}$.
It is easy to confirm the following lemma.
\begin{mylemma}[\cite{v001a005}]
\begin{equation*}
\bigl|\bra{0}HR_z(\theta)H\ket{0}\bigr|^2
= \frac{1+\cos(\theta)}2.
\end{equation*}
\end{mylemma}
From this lemma, the bias in~\eqref{eq:werner} is obtained by the following one-qubit circuit.
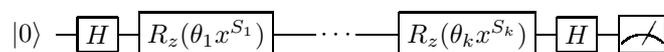
\[
\Qcircuit @C=.8em @R=0em @!R {
\lstick{\ket{0}} & \gate{H} & \gate{R_z(\theta_1 x^{S_1})} & \qw & \push{\,\dotsm\,} \qw & \qw & \gate{R_z(\theta_k x^{S_k})}& \gate{H} & \meter
}
\]
where $x^S:=\prod_{i\in S}x_i$.
By using fan-out gates~\cite{green2002counting}, \cite{v001a005}, we obtain the following equivalent circuit.
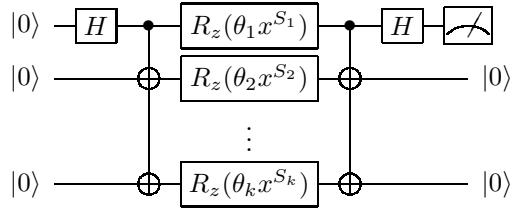
\[
\Qcircuit @C=.8em @R=0.3em @!R {
\lstick{\ket{0}}& \gate{H} & \ctrl{3} & \gate{\textover[c]{$R_z(\theta_1 x^{S_1})$}{$R_z(\theta_k x^{S_k})$}} &  \ctrl{3} & \gate{H} & \meter\\
\lstick{\ket{0}} & \qw      & \targ    & \gate{\textover[c]{$R_z(\theta_2 x^{S_2})$}{$R_z(\theta_k x^{S_k})$}} &  \targ    & \qw & \rstick{\ket{0}} \qw \\
 &       &     & \vdots &     &  &  \\
\lstick{\ket{0}} & \qw      & \targ    & \gate{R_z(\theta_k x^{S_k})} &  \targ    & \qw & \rstick{\ket{0}} \qw\\
}
\]
This circuit is useful when we consider computations by \QNCfz circuit~\cite{v001a005}, \cite{Takahashi2016}.
Since the fan-out gate is equivalent to the Mod2 gate conjugated by Hadamard gates, we obtain the following equivalent circuit.
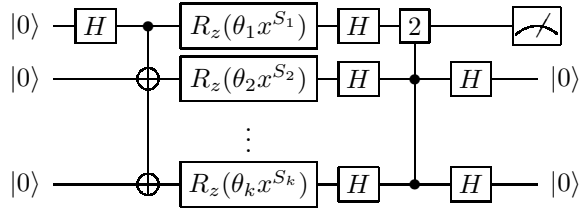
\[
\Qcircuit @C=.8em @R=0.3em @!R {
\lstick{\ket{0}}& \gate{H} & \ctrl{3} & \gate{\textover[c]{$R_z(\theta_1 x^{S_1})$}{$R_z(\theta_k x^{S_k})$}} &  \gate{H} & \gate{2} & \qw & \meter\\
\lstick{\ket{0}} & \qw      & \targ    & \gate{\textover[c]{$R_z(\theta_2 x^{S_2})$}{$R_z(\theta_k x^{S_k})$}} &  \gate{H} & \ctrl{-1}    & \gate{H} & \rstick{\ket{0}} \qw \\
 &       &     & \vdots &     &  &  \\
\lstick{\ket{0}} & \qw      & \targ    & \gate{R_z(\theta_k x^{S_k})} &  \gate{H} & \ctrl{-2}    & \gate{H} & \rstick{\ket{0}} \qw\\
}
\]
Here, we can measure all qubits before the Mod2 gate is applied and output XOR of measurement outcomes.
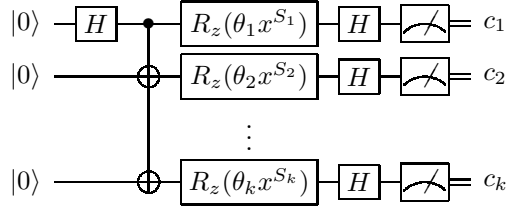
\[
\Qcircuit @C=.8em @R=0.3em @!R {
\lstick{\ket{0}}& \gate{H} & \ctrl{3} & \gate{\textover[c]{$R_z(\theta_1 x^{S_1})$}{$R_z(\theta_k x^{S_k})$}} &  \gate{H} & \meter & \rstick{c_1}\cw\\
\lstick{\ket{0}} & \qw      & \targ    & \gate{\textover[c]{$R_z(\theta_2 x^{S_2})$}{$R_z(\theta_k x^{S_k})$}} &  \gate{H} & \meter & \rstick{c_2}\cw\\
 &       &     & \vdots &     &  &  \\
\lstick{\ket{0}} & \qw      & \targ    & \gate{R_z(\theta_k x^{S_k})} &  \gate{H} &  \meter & \rstick{c_k}\cw
}
\]
The last circuit is equivalent to a setting of \NMQCp using shared generalized GHZ state. The first Hadamard gate and fan-out gate generate the generalized GHZ state and following gates and measurements corresponding to
non-adaptive measurements of qubits.
Since the computation needs only one qubit on a circuit model, it can be simulated by $\TC^0$ circuit as mentioned in~\cite{green2002counting}.
Note that we can also directly prove this by simulating periodicity of cosine function by threshold circuits.
Computational powers of several types of MBQC are summarized in Table~\ref{tbl:mbqc}.

\begin{table}[t]
\caption{
Computational power of measurement-based quantum computation with linear side-processing.
}
\label{tbl:mbqc}
\begin{center}
\begin{tabular}{|c|c|c|c|}
\hline
Side processor & State & Adaptivity & Complexity\\
\hline
Linear & Cluster state & Adaptive & \BQP~\cite{PhysRevA.68.022312}\\
Linear & Tripartite GHZ state & Adaptive & \P~\cite{PhysRevLett.102.050502}\\
Linear & Generalized GHZ state & Non-adaptive & $\subseteq \TC^0$\\
\hline
\end{tabular}
\end{center}
\end{table}

\section{Fourier representation of $\mathsf{Mod}^3_n$}\label{apx:mod3}
\begin{mylemma}
For $n$ divisible by 3, $\widehat{\mathsf{Mod}^{3}_{n}}(S)$ is in $\mathbb{R}_{n-2}\setminus\mathbb{R}_{n-3}$ for all $S\subseteq[n]$ of even size and is equal to zero for all $S$ of odd size.
For $n$ not divisible by 3, $\widehat{\mathsf{Mod}^{3}_{n}}(S) \in \mathbb{R}_{n-1}\setminus\mathbb{R}_{n-2}$ for all non-empty $S\subseteq[n]$.
\end{mylemma}
\begin{proof}
Assume that $n$ is divisible by 3.
Then, $2\mathrm{Re}\left(\prod_{i=1}^n\left(\frac{-1+x_i\sqrt{-3}}2\right)\right)$ is 2 if $\sum_{i\in[n]}\mathsf{x}_i$ is divisible by 3, and -1 if not where $\mathrm{Re}(\omega)$ is the real part of a complex number $\omega$.
Hence,
\begin{align*}
\mathsf{Mod}^{3}_n(x)
&= \frac13\left(-4\mathrm{Re}\left(\prod_{i=1}^n\left(\frac{-1+x_i\sqrt{-3}}2\right)\right) + 1\right)\\
&= \frac13\left(\frac{(-1)^{n+1}}{2^{n-2}}\mathrm{Re}\left(\sum_{S\subseteq[n]} \prod_{i\in S}(-x_i\sqrt{-3})\right) + 1\right)\\
&= \frac13\left(\frac{(-1)^{n+1}}{2^{n-2}}\sum_{\substack{S\subseteq[n]\\ |S|\text{ is even}}} \prod_{i\in S}(x_i\sqrt{-3}) + 1\right)\\
&= \frac13 +\frac{(-1)^{n+1}}{2^{n-2}3}\sum_{\substack{S\subseteq[n]\\ |S|\text{ is even}}} (-3)^{|S|/2}\prod_{i\in S}x_i.
\end{align*}
Hence, Fourier coefficients for $S$ of odd size are zero and those for $S$ of even size are in $\mathbb{R}_{n-2}\setminus\mathbb{R}_{n-3}$.

Assume that $n\equiv 1 \mod 3$.
Then, $2\mathrm{Re}\left(\left(\frac{-1-\sqrt{-3}}2\right)\prod_{i=1}^n\left(\frac{-1+x_i\sqrt{-3}}2\right)\right)$ is 2 if $\sum_{i\in[n]}\mathsf{x}_i$ is divisible by 3, and -1 if not.
Hence,
\begin{align*}
\mathsf{Mod}^{3}_n(x)
&= \frac13\left(-4\mathrm{Re}\left(\left(\frac{-1-\sqrt{-3}}2\right)\prod_{i=1}^n\left(\frac{-1+x_i\sqrt{-3}}2\right)\right) + 1\right)\\
&= \frac13\left(\frac{(-1)^{n}}{2^{n-1}}\mathrm{Re}\left((1+\sqrt{-3})\sum_{S\subseteq[n]} \prod_{i\in S}(-x_i\sqrt{-3})\right) + 1\right)\\
&= \frac13\left(\frac{(-1)^n}{2^{n-1}}\left(\sum_{\substack{S\subseteq[n]\\ |S|\text{ is even}}} \prod_{i\in S}(x_i\sqrt{-3}) - \sqrt{-3} \sum_{\substack{S\subseteq[n]\\ |S|\text{ is odd}}} \prod_{i\in S}(x_i\sqrt{-3})\right) + 1\right)\\
&= \frac13 +\frac{(-1)^n}{2^{n-1}3}\sum_{S\subseteq[n]} (-1)^{\lfloor |S|/2 \rfloor} 3^{\lfloor(|S|+1)/2\rfloor}\prod_{i\in S}x_i
\end{align*}
Hence, all Fourier coefficients for non-empty $S$ are in $\mathbb{R}_{n-1}\setminus\mathbb{R}_{n-2}$.

Assume that $n\equiv 2 \mod 3$.
Then, $2\mathrm{Re}\left(\left(\frac{-1+\sqrt{-3}}2\right)\prod_{i=1}^n\left(\frac{-1+x_i\sqrt{-3}}2\right)\right)$ is 2 if $\sum_{i\in[n]}\mathsf{x}_i$ is divisible by 3, and -1 if not.
The rest of the proof is same as that for the case $n\equiv 1\mod 3$.
\end{proof}

\section{LSB function: The proof of Lemma~\ref{lem:lsb}}\label{apx:lsb}
\begin{figure}[t]
\begin{center}
\ttfamily
\begin{tabular}{r@{\hspace{5em}}r}
0001 & 110101001110\\
0010 & 110101001111\\
0011 & 110101010000\\
\vdots & \vdots \\
1111 & 110101011100\\
10000 & 110101011101
\end{tabular}
\end{center}
\caption{
Let $n$ be $\mathtt{110101011101}$ as the binary representation and $\ell=5$.
Left: Numbers from 1 to $2^{\ell-1}$. Right: Numbers from $n-2^{\ell-1}+1$ to $n$.}
\label{fig:lsb}
\end{figure}
From
\begin{align*}
\bigoplus_{1\le i_1 < i_2 < \dotsb < i_{2^{\ell-1}}\le n} x_{i_1}\land x_{i_2}\land \dotsm \land x_{i_{2^{\ell-1}}} &= \binom{n}{2^{\ell-1}} \mod 2\\
&=\frac{n(n-1)\dotsm(n-2^{\ell-1}+1)}{2^{\ell-1}!} \mod 2
\end{align*}
we will count the numbers of factor 2s in the denominator and the numerator.
The number of factor 2s in arbitrary given integer $k$ is equal to the number of trailing zeros in the binary representation of $k$.
Hence, the number of factor 2s in $2^{\ell-1}!$ is equal to the sum of the number of trailing zeros of integers from 1 to $2^{\ell-1}$.
We do not have to count it explicitly.
Next, we count the number of factor 2s in the numerator.
In integers from $n - 2^{\ell-1} + 1$ to $n$, all of bit patterns of length $\ell-1$ appear in the least significant $\ell-1$ bits as shown in Fig.~\ref{fig:lsb}.
An important case is the all-zero case.
Let $m$ be the unique integer at least $n - 2^{\ell-1} + 1$ and at most $n$ which can be divided by $2^{\ell-1}$.
For counting the number of trailing zeros of $m$, the $\ell$-th bit of $m$ is concerned.
Here, the $\ell$-th bit of $m$ is equal to the $\ell$-th bit of $n$.
If the $\ell$-th bit of $n$ is 1, then, the number of trailing zeros of $m$ is equal to $\ell-1$.
In this case the numbers of factor 2s in the numerator and the denominator are equal.
If the $\ell$-th bit of $n$ is 0, then, the number of trailing zeros of $m$ is greater than $\ell-1$.
In this case, the number of factor 2s in the numerator is greater than that in the denominator.
This means that $\binom{n}{2^{\ell-1}} \mod 2$ is equal to $\mathrm{LSB}^\ell(\sum_i x_i)$.
An example is shown in Fig.~\ref{fig:lsb}.
\end{appendices}

\end{document}